\documentclass[fleqn,runningheads]{llncs}

\usepackage{hyperref}
\usepackage{latexsym}
\usepackage{amssymb}
\usepackage{amsmath}
\usepackage{xcolor}

\definecolor{darkgreen}{RGB}{0,128,0}
\newcommand{\red}[1]{{\color{red} #1}}
\newcommand{\green}[1]{{\color{darkgreen} #1}}
\newcommand{\blue}[1]{{\color{blue} #1}}
\newcommand{\olive}[1]{{\color{olive} #1}}

\newcommand{\his}{\Pi} 

\newcommand{\id}{{\mathit{id}}} %% empty substitution
\newcommand{\true}{{\mathtt{true}}} %% empty clause

\newcommand{\sleq}{\leqslant}
\newcommand{\ri}{{<\!\!\!<}}

\newcommand{\body}{\mathsf{body}}
\newcommand{\clauses}{\mathsf{clauses}}

\def\defemb#1#2{\expandafter\def\csname #1\endcsname
                              {\relax\ifmmode #2\else\hbox{$#2$}\fi}}
\defemb{cA}{{\cal A}}
\defemb{cB}{{\cal B}}
\defemb{cC}{{\cal C}}
\defemb{cD}{{\cal D}}
\defemb{cH}{{\cal H}}
\defemb{cT}{{\cal T}}
\defemb{cV}{{\cal V}}

\newcommand{\var}{{\cV}ar}
\newcommand{\dom}{{\cD}om}
\newcommand{\nil}{[\:]}

\newcommand{\ol}[1]{\overline{#1}}  % sequence of objects

\newcommand{\ott}{\;\mbox{\tt :-}\;}

\def\ri{<\!\!\!<}   % renamed instance

\def\res{\mathrel{\vert\grave{ }}}

\def \tuple#1{\langle #1 \rangle}

\long\def\comment#1{}

\newcommand{\fail}{\mathsf{fail}}
\newcommand{\mgu}{\mathsf{mgu}}

\newcommand{\midd}{\!\mid\!}
\newcommand{\sep}{\bullet}

\newcommand{\cons}{\!:\!}
\newcommand{\fto}{\rightharpoonup}
\newcommand{\bto}{\leftharpoondown}

\newcommand{\choice}{\mathsf{ch}}
\newcommand{\call}{\mathsf{call}}
\newcommand{\choicefail}{\mathsf{fail}}
\newcommand{\backtrack}{\mathsf{bck}}
\newcommand{\ret}{\mathsf{ret}}
\newcommand{\exit}{\mathsf{exit}}
\newcommand{\redo}{\mathsf{redo}}
\newcommand{\nextsol}{\mathsf{next}}
\newcommand{\unf}{\mathsf{unf}}
\newcommand{\clause}{\mathsf{cl}}

\begin{document}

\title{Reversible Debugging in Logic Programming%
  \thanks{This work has been partially supported by grant PID2019-104735RB-C41
funded by MCIN/AEI/ 10.13039/501100011033, by the 
\emph{Generalitat Valenciana} under grant Prometeo/2019/098 
(DeepTrust), and by the COST Action 
  IC1405 on Reversible Computation - extending horizons of computing.}
}

\author{Germ\'an Vidal \orcidID{0000-0002-1857-6951}} %\inst{1}}

\institute{
MiST, VRAIN, Universitat Polit\`ecnica de Val\`encia, Spain\\
  \email{gvidal@dsic.upv.es}
}

\maketitle

\begin{abstract}
	Reversible debugging is becoming increasingly popular for
	locating the source of errors. This technique
	proposes a more natural approach to debugging, where one can
	explore a computation from the observable misbehaviour 
	\emph{backwards} to the source of the error. 
	In this work, we propose a reversible debugging scheme for
	logic programs. For this purpose, we define an appropriate 
	instrumented semantics (a so-called Landauer embedding)
	that makes SLD resolution reversible. An
	implementation of a reversible debugger for Prolog, 
	\textsf{rever},
    has been developed and is publicly available.\\[2ex]
    \emph{This paper is dedicated to Manuel Hermenegildo
    on his 60th birthday, 
    for his many contributions to logic programming as well as 
    his energetic leadership within the community. I wish him 
    many springs more to come.}
\end{abstract}

\section{Introduction}
\label{sec:intro}

Reversible debugging allows one to explore a program execution
back and forth. In particular, if one observes a misbehaviour 
in some execution (e.g., a variable that 
takes a wrong value or an unexpected exception), reversible
debugging allows us to analyse the execution backwards 
from this point. This feature is
particularly useful for long executions, where a step-by-step
forward inspection from the beginning of the execution would 
take too much time, or be even impractical.

One can already find a number of tools for reversible debugging
in different programming languages, like Undo \cite{UndoWhitePaper}, 
rr \cite{OJFHNP17} or CauDEr \cite{LPV19}, to name a few.
In this work, we consider reversible debugging in \emph{logic
programming} \cite{Llo87}. In this context, one has to deal 
with two specific features that are not common in other 
programming languages: nondetermism and a bidirectional 
parameter passing mechanism (unification). 

Typically, the \emph{reversibilization} of a (reduction) semantics
can be obtained by instrumenting the states with an appropriate
\emph{Landauer embedding} \cite{Lan61}, i.e., by introducing
a \emph{history} where the information required to undo the
computation steps is stored. 
Defining a Landauer embedding for logic programming is a
challenging task because of nondetermism and unification.
On the one hand, in order to undo backtracking steps, a
deterministic semantics that models the complete traversal
of an SLD tree is required (like the linear operational semantics
introduced in \cite{SESGF11}). On the other hand,
unification is an irreversible operation: given two terms,
$s$ and $t$, with most general unifier $\sigma$, we cannot
obtain $s$ from $t$ and $\sigma$ (nor $t$ from $s$ and $\sigma$).

Let us note that, in this work, we aim at reversibility 
in the sense of
being able to \emph{deterministically} undo the steps of
a computation. In general, (pure) logic programs are 
\emph{invertible} (e.g., the same relation can be used
for both addition and subtraction), but they are not
reversible in the above sense.

This paper extends the preliminary results reported in the
short paper \cite{Vid20}. In particular, our main contributions
are the following:
\begin{itemize}
  \item First, we define a reversible operational semantics
  for logic programs that deals explicitly with backtracking
  steps. In particular, we define both a \emph{forward} and
  a \emph{backward} transition relation that model forward and 
  backward computations, respectively.
  \item Moreover, we state and prove some formal properties
  for our reversible semantics, including the fact that it
  is indeed a conservative extension of the standard
  semantics, that it is deterministic, and that any
  forward computation can be undone.
  \item Finally, we present the design of a reversible
  debugger for Prolog that is based on our reversible
  semantics, and discuss some aspects of the implemented
  tool, the reversible debugger \textsf{rever}.
\end{itemize}
We consider that our work can be useful in the context
of existing techniques for program validation in
logic programming, like run-time verification
(e.g., \cite{SMH14}) or concolic testing 
(e.g., \cite{MPV15}), in order to help locating 
the bugs of a program.

The paper is organised as follows. After introducing some
preliminaries in the next section, we introduce our reversible
operational semantics in Section~\ref{sec:reversible}.
Then, Section~\ref{sec:rever} presents the design of a
reversible debugger based on the previous semantics. Finally, 
Section~\ref{sec:relwork} compares our approach with some
related work and Section~\ref{sec:conclusion} concludes
and points out some directions for further research.

%%%%%%%%%%%%%%%%%%%%%%%%%%%%%%%%%%%%%%%%%%%%%%%%%%%%%%%%%%%%%%%%%%%
\section{Preliminaries}

In this section, we 
briefly recall some basic notions from logic programming
(see, e.g., \cite{Llo87,Apt97} for more details).  

In this work, we consider a first-order language with a fixed
vocabulary of predicate symbols, function symbols, and variables
denoted by $\Pi$, $\Sigma$ and $\cV$, respectively, 
with $\Sigma\cap\Pi=\emptyset$
and $(\Sigma\cup\Pi)\cap\cV=\emptyset$.
Every element of $\Sigma\cup\Pi$ has an \emph{arity} which is the number
of its arguments. We write $f/n\in \Sigma$ 
(resp.\ $p/n\in \Pi$) to denote
that $f$ (resp.\ $p$) is an element of $\Sigma$ (resp.\ $\Pi$) 
whose arity is $n\geq 0$. A \emph{constant symbol} is an element 
of $\Sigma$ whose arity is 0.
We let $\cT(\Sigma,\cV)$ denote the set of \emph{terms} constructed
using symbols from $\Sigma$ and variables from $\cV$.  

An \emph{atom} has the form $p(t_1,\ldots,t_n)$ with 
$p/n \in \Pi$ and $t_i
\in \cT(\Sigma,\cV)$ for $i = 1,\ldots,n$.  
%
%The notion of position is extended to atoms in the natural way.  
A \emph{query} is a
finite conjunction of atoms which is denoted by a sequence of the form
$A_1,\ldots, A_n$, where the \emph{empty query} is denoted by $\true$.
A \emph{clause} has the form $H\leftarrow B_1,\ldots, B_n$, where $H$
(the \emph{head}) and $B_1,\ldots,B_n$ (the \emph{body}) are atoms,
$n\geq 0$ (thus we only consider \emph{definite} logic programs, i.e.,
logic programs without negated atoms in the body of the
clauses). Clauses with an empty body, $H\leftarrow \true$, are called
\emph{facts}, and are typically denoted by $H$.
In the following, atoms are ranged over by $A,B,C,H,\ldots$ while
queries (possibly empty sequences of atoms) are ranged over by
$\cA,\cB,\ldots$

$\var(s)$ denotes the set of variables in the syntactic object $s$
(i.e., $s$ can be a term, an atom, a query, or a clause).  A syntactic
object $s$ is \emph{ground} if $\var(s)=\emptyset$. In this work, we
only consider \emph{finite} ground terms.

Substitutions and their operations are defined as usual; 
they are typically denoted by (finite) sets of bindings like,
e.g., $\{x_1/s_1,\ldots,x_n/ s_n\}$. We let $\id$ denote 
the identity substitution. Substitutions are
ranged over by $\sigma,\theta,\ldots$
In particular, the set $\dom(\sigma) = \{x \in \cV \mid \sigma(x) \neq
x\}$ is called the \emph{domain} of a substitution $\sigma$. 
Composition of substitutions is denoted by juxtaposition,
i.e., $\sigma\theta$ denotes a substitution $\gamma$ such that
$\gamma(x) = \theta(\sigma(x))$ for all $x\in\cV$. We
follow a postfix notation for substitution application: given a
syntactic object $s$ and a substitution $\sigma$ the application 
$\sigma(s)$ is denoted by $s\sigma$.
The \emph{restriction} $\theta\!\res_V$ of a substitution $\theta$ to a
set of variables $V$ is defined as follows: $x\theta\!\res_{V} =
x\theta$ if $x\in V$ and $x\theta\!\res_V = x$ otherwise. We say that
$\theta = \sigma~[V]$ if $\theta\!\res_V = \sigma\!\res_V$.

A syntactic object $s_1$ is \emph{more general} than a syntactic
object $s_2$, denoted $s_1 \sleq s_2$, if there exists a substitution
$\theta$ such that $s_2 = s_1\theta$. A \emph{variable renaming} is a
substitution that is a bijection on $\cV$. Two syntactic objects $t_1$
and $t_2$ are \emph{variants} (or equal up to variable renaming),
denoted $t_1 \approx t_2$, if $t_1 = t_2\rho$ for some variable renaming
$\rho$. A substitution $\theta$ is a unifier of two syntactic objects
$t_1$ and $t_2$ iff $t_1\theta = t_2\theta$; furthermore, $\theta$ is
the \emph{most general unifier} of $t_1$ and $t_2$, denoted by
$\mgu(t_1,t_2)$ if, for every other unifier $\sigma$ of $t_1$ and
$t_2$, we have that $\theta \sleq \sigma$. 

A logic \emph{program} is a finite sequence of clauses.  Given a
program $P$, we say that $A,\cB' \leadsto_{P,\sigma} (\cB,\cB')\sigma$
is an \emph{SLD resolution step}\footnote{In this paper, we only
  consider Prolog's \emph{computation rule}, so that the selected atom
  in a query is always the leftmost one.} if $H \leftarrow \cB$ is a
renamed apart clause (i.e., with fresh variables) of program $P$, in
symbols, $H \leftarrow \cB \ri P$, and $\sigma = \mgu(A,H)$.
% is the \emph{most
% general unifier} of $A$ and $H$.
The subscript $P$ will often be omitted when the program is clear from
the context.  An \emph{SLD derivation} is a (finite or infinite)
sequence of SLD resolution steps. 
As is common, $\leadsto^\ast$ denotes the reflexive and 
transitive closure of $\leadsto$. In particular, we denote by
$\cA_0 \leadsto^\ast_\theta \cA_n$ a derivation 
\[
\cA_0 \leadsto_{\theta_1} \cA_1 \leadsto_{\theta_2}
\ldots \leadsto_{\theta_n} \cA_n
\]
where $\theta =
\theta_1\ldots\theta_n$ if $n>0$ (and $\theta =\id$ otherwise).

An SLD derivation is called \emph{successful} if it ends with 
the query $\true$, and it is called
\emph{failed} if it ends in a query where the leftmost atom 
does not unify with the head of any clause. 
Given a successful SLD derivation
$\cA \leadsto^\ast_\theta \true$, the associated \emph{computed
answer}, $\theta\!\res_{\var(\cA)}$, is the restriction of $\theta$
to the variables of the initial query $\cA$.
SLD derivations are represented by a (possibly infinite) finitely
branching tree, which is called \emph{SLD tree}. Here, \emph{choice
  points} (queries with more than one child) correspond to queries
where the leftmost atom unifies with the head of more than one program
clause.

\begin{example} \label{ex:running}
Consider the following (labelled)
 logic program:\footnote{We consider Prolog notation in examples (so variables start with an uppercase letter). Clauses are labelled 
with a unique identifier of the form $\ell_i$.}
\[
  \begin{array}{l@{~}l@{~~~~~~~~}l@{~}l}
  	\ell_1: & \tt p(X,Y) \ott q(X),r(X,Y). \\%& (r_1)\\
  	\ell_2: & \tt q(a). %& (r_2) 
  	& \ell_5: & \tt r(b,b). \\%& (r_5)\\
  	\ell_3: & \tt q(b). %& (r_3) 
  	& \ell_6: & \tt r(b,c). \\%& (r_6)\\
  	\ell_4: & \tt q(c). %& (r_4) 
  	& \ell_7: & \tt r(c,c). \\%& (r_7)\\
  \end{array}
\]
Given the query $\tt p(X,Y)$, we have, e.g.,  the following
(successful) SLD derivation:
\[
  	\begin{array}{lll}
  		\tt p(A,B) 
  		& \leadsto_{\tt \{X/A,Y/B\}} & \tt q(A),r(A,B) \\
  		& \leadsto_{\tt \{A/b\}} & \tt r(b,B) \\
  		& \leadsto_{\tt \{B/c\}} & \tt \true \\
  	\end{array}
\]
with computer answer $\tt \{A/b,B/c\}$.
\end{example}

%%%%%%%%%%%%%%%%%%%%%%%%%%%%%%%%%%%%%%%%%%%%%%%%%%%%%%%%%%%
\section{A Reversible Semantics for Logic Programs} 
\label{sec:reversible}

In this section, we present a reversible semantics for logic programs
that constitutes a good basis to implement a reversible 
debugger for Prolog (cf.~Section~\ref{sec:rever}).
In principle, one of the main challenges for
defining a reversible version of SLD resolution is dealing
with unification, since it is an irreversible operation.
E.g., given the SLD resolution step
\[
\tt p(X,a),q(a) \leadsto_{\{X/a,Y/a\}} q(a),q(a)
\]
using clause $\tt p(a,Y) \ott q(Y)$, 
there is no deterministic
way to get back the query $\tt p(X,a),q(a)$ from 
the query $\tt q(a),q(a)$, the computed 
mgu $\tt \{X/a,Y/a\}$, and the applied clause. 
For instance, one could obtain the query
$\tt p(X,X),q(X)$ since the following SLD resolution step
\[
\tt p(X,X),q(X) \leadsto_{\{X/a,Y/a\}} q(a),q(a)
\]
is also possible using the same clause and computing the 
same mgu.

In order to overcome this problem, \cite{Vid20} proposed
a reversible semantics where 
\begin{itemize}
  \item computed mgu's are not applied 
to the atoms of the query, and
  \item the selected call at each SLD resolution step
  is also stored.
\end{itemize}
Queries are represented as pairs  
$\tuple{\cA;[(A_n,H_n,m_n),\ldots,(A_1,H_1,m_1)]}$, 
where the first component is a sequence of atoms
(a query), and the
second component stores, for each SLD resolution step
performed so far,
the selected atom ($A_i$), the head of the selected clause ($H_i$),
and the number of atoms in the body of this clause ($m_i$).
Here, mgu's are not stored explicitly but can be inferred
from the pairs $(A_i,H_i)$. The number $m_i$ is used to
determine the number of atoms in the current query that
must be removed when performing a backward step.
A reversible SLD resolution step has then the form\footnote{Here, 
$(A,H,m)\cons\cH$
denotes a list with head element $(A,H,m)$ and tail $\cH$.} 
\[
\tuple{A,\cB;\cH} \fto
\tuple{B_1,\ldots,B_m,\cB;
(A,H,m)\cons\cH}
\]
if there exists a clause $H \leftarrow B_1,\ldots,B_m\ri P$
and $\mgu(A\sigma,H) \neq \fail$, where $\sigma$ is the substitution
obtained from $\cH$ by computing the mgu's associated to
each triple ($A_i,H_i,m)$ in $\cH$ and, then, composing them.
A simple proof-of-concept implementation that follows this
scheme can be found at
\url{https://github.com/mistupv/rever/tree/rc2020}.

The proposal in \cite{Vid20}, however, suffers from
several drawbacks:
\begin{itemize}
  \item First, it is very inefficient, since one should
  compute the mgu's of each SLD resolution step once and
  again. This representation was chosen in \cite{Vid20} for
  clarity and, especially, because it allowed us to easily implement 
  it without
  using a ground representation for queries and programs,
  so that there was no need to reimplement all basic 
  operations (mgu, substitution application and
  composition, etc).
  
  \item The second drawback is that the above definition 
  of reversible SLD resolution
  cannot be used to undo a \emph{backtracking} step, since the
  structure of the SLD tree is not explicit in the 
  considered semantics. 
\end{itemize}
In the following, we introduce a reversible operational 
semantics for logic programs that overcomes the above
shortcomings.

%%%%%%%%%%%%%%%%%%%%%%%%%%%%%%%%%%%%%%%%%%%%%%%%%%%%%%%%%%
\subsection{A Deterministic Operational Semantics}

First, we present a deterministic semantics (inspired by the
linear operational semantics of \cite{SESGF11}) that deals
explicitly with backtracking.

Our semantics is defined as a transition relation on
states. In the following, queries are represented
as pairs $\tuple{\cA;\theta}$ instead of $\cA\theta$,
where $\theta$ is the composition of the mgu's computed
so far in the derivation. This is needed in order to
avoid undoing the application of mgu's, which is
an irreversible operation. 

\begin{definition}[state] \label{def:state}
  A \emph{state} is denoted by a sequence 
  $Q_1\midd Q_2\midd \ldots \midd Q_n$, where each
  $Q_i$ is a (possibly labelled) query of the form 
  $\tuple{\cB;\theta}$.
  In some cases, a query $Q$ is labelled with a clause label, e.g.,
   $Q^\ell$, which will be used to denote that the query $Q$ 
  can be unfolded with the clause labelled with $\ell$ (see below). 
\end{definition}
A state will often
be denoted by $\tuple{\cB;\theta}\midd S$ so that 
$\tuple{\cB;\theta}$ 
is  the first query of the sequence and $S$ denotes 
a (possibly empty) sequence of queries. 
In the following, an empty sequence is denoted by $\epsilon$. 

\begin{figure}[t]
  \[
  \hspace{-2ex}\begin{array}{r@{~}l}
    \mathsf{(backtrack)} & {\displaystyle 
      \frac{S\neq\epsilon} 
        {\tuple{\fail,\cB;\theta} \midd S \to 
        S
        }}
        \hspace{3ex}

    \mathsf{(next)} ~ {\displaystyle 
      \frac{S\neq\epsilon} 
        {\tuple{\true;\theta}\midd S \to S
        }}\\[5ex]

     \textsf{(choice)} &  {\displaystyle 
      \frac{
      A\neq\fail\wedge A\neq\true \wedge \clauses(A\theta,P) = \{\ell_1,\ldots,\ell_m\}\wedge m>0
      } 
        {\tuple{A,\cB;\theta} \midd S 
        \to
          \tuple{A,\cB;\theta}^{\ell_1} \midd 
          \ldots\midd
          \tuple{A,\cB;\theta}^{\ell_m}\midd  S
        }
        } \\[4ex]

     \textsf{(choice\_fail)} &  {\displaystyle 
      \frac{
      A\neq\fail\wedge A\neq\true \wedge \clauses(A\theta,P) = \emptyset} 
        {\tuple{A,\cB;\theta}\midd S 
        \to 
          \tuple{\fail,\cB;\theta}\midd  S}
        } \\[4ex]

    \mathsf{(unfold)} & {\displaystyle 
      \frac{\clause(\ell,P) = H\leftarrow B_1,\ldots,B_n
 	\wedge \mgu(A\theta,H)=\sigma}
        {\tuple{A,\cB;\theta}^{\ell}
        \midd S
    \to \tuple{B_1,\ldots,B_n,\cB;\theta\sigma}\midd S}
        } 
    \end{array}
    \]
  \caption{A deterministic operational semantics} 
  \label{fig:deterministic}
\end{figure}

In this paper, we consider that program clauses are labelled,
so that each label uniquely identifies a program clause. Here, we
use the auxiliary function $\clauses(A,P)$ to obtain the labels
of those clauses in program $P$ whose heads unify with atom $A$,
i.e., 
\[
\clauses(A,P) = \{\ell \mid \ell:H\leftarrow \cB\ri P
\wedge \mgu(A,H)\neq \fail\}
\]
and $\clause(\ell,P)$ to get a
renamed apart variant of the clause labelled with $\ell$, i.e.,
$\clause(\ell,P) = (H\leftarrow\cB)\vartheta$ if 
$\ell:H\leftarrow\cB \in P$ and $\vartheta$ is a variable
renaming with fresh variables.

The rules of the semantics can be found in 
Figure~\ref{fig:deterministic}. An \emph{initial state}
has the form $\tuple{A,\cB;\id}$, where $A$ is an atom,
$\cB$ is a (possibly empty) sequence of atoms,
and $\id$ is the identity substitution. 
Initially, one can either apply rule \textsf{choice} or
\textsf{choice\_fail}. Let us assume that $A$ unifies with 
the head of some clauses, say $\ell_1,\ldots,\ell_m$. 
Then, rule \textsf{choice} derives a new state by replacing
$\tuple{A,\cB;\id}$ with $m$ copies
 labelled with $\ell_1,\ldots,\ell_m$: 
 \[
\tuple{A,\cB;\id} \to  \tuple{A,\cB;\id}^{\ell_1}\midd\ldots
 \midd\tuple{A,\cB;\id}^{\ell_m}
 \]
Now, let assume that $\clause(\ell_1,P)$ returns
$H\leftarrow B_1,\ldots,B_n$. Then, rule \textsf{unfold}
applies so that the following state is derived:
\[
\tuple{B_1,\ldots,B_n,\cB;\sigma}\midd \tuple{A,\cB;\id}^{\ell_2}
\midd \ldots\midd \tuple{A,\cB;\id}^{\ell_m}
\]
Let us consider now that $B_1\sigma$ does not match 
any program clause, i.e., we have $\clauses(B_1\sigma,P) = \emptyset$.
Then, rule \textsf{choice\_fail} applies
and the following state is derived: 
\[
\tuple{\fail,B_2,\ldots,B_n,\cB;\sigma}\midd 
\tuple{A,\cB;\id}^{\ell_2}\midd \ldots\midd 
\tuple{A,\cB;\id}^{\ell_m}
\]
Then, rule \textsf{backtrack} applies so that we jump
to a choice point with some pending alternative (if any).
In this case, we derive the state
\[
\tuple{A,\cB;\id}^{\ell_2}\midd \ldots
\midd \tuple{A,\cB;\id}^{\ell_m}
\]
so that unfolding with clause $\ell_2$ is tried now,
and so forth. 

Here, we say that a derivation is 
\emph{successful} if the last state has the
form $\tuple{\true;\theta}\midd S$. 
We have also included a rule called \textsf{next} to
be able to reach all solutions of an SLD tree (which 
has a similar effect as rule \textsf{backtrack}). 
Therefore, $\theta$ is not necessarily the first
computed answer, but an arbitrary one (as long as
it is reachable from the initial state after a finite
number of steps).

A computation is \emph{failed} if it ends with 
a state of the form $\tuple{\fail,\cB;\theta}$, so
no rule is applicable (note that rule \textsf{backtrack}
is not applicable when there is a single query in 
the state). 

\begin{example} \label{ex:running2}
Consider the program of Example~\ref{ex:running} and the
same initial query: $\tt \tuple{p(X,Y);\id}$. In order to 
reach the same computed answer, $\tt \{A/b,B/c\}$, we now
perform the following (deterministic) 
derivation:\footnote{For clarity, we only show the bindings
for the variables in the initial query. Moreover, the steps
are labelled with the applied rule.
%, and we use different 
%colours for the atoms of the query and the current
%substitution.
}
\[
\begin{array}{lll}
 \tt \tuple{p(A,B);\red{\id}}
 & \to_{\sf choice} & \tt \tuple{p(A,B);\red{\id}}^{\ell_1}  \\
 & \to_{\sf unfold} & \tt \tuple{q(A),r(A,B);\red{\id}}\\  
 & \to_{\sf choice} & \tt \tuple{q(A),r(A,B);\red{\id}}^{\ell_2}
 \midd \tuple{q(A),r(A,B);\red{\id}}^{\ell_3}\midd
 \tuple{q(A),r(A,B);\red{\id}}^{\ell_4}\\  
 & \to_{\sf unfold} & \tt \tuple{r(A,B);\red{\{A/a\}}}
 \midd \tuple{q(A),r(A,B);\red{\id}}^{\ell_3}\midd
 \tuple{q(A),r(A,B);\red{\id}}^{\ell_4}\\  
 & \to_{\sf choice\_fail} & \tt \tuple{\fail;\red{\{A/a\}}}
 \midd \tuple{q(A),r(A,B);\red{\id}}^{\ell_3}\midd
 \tuple{q(A),r(A,B);\red{\id}}^{\ell_4}\\  
 & \to_{\sf backtrack} & \tt \tuple{q(A),r(A,B);\red{\id}}^{\ell_3}\midd
 \tuple{q(A),r(A,B);\red{\id}}^{\ell_4}\\  
 & \to_{\sf unfold} & \tt \tuple{r(A,B);\red{\{A/b\}}}\midd
 \tuple{q(A),r(A,B);\red{\id}}^{\ell_4}\\  
 & \to_{\sf choice} & \tt \tuple{r(A,B);\red{\{A/b\}}}^{\ell_5}\midd
 \tuple{r(A,B);\red{\{A/b\}}}^{\ell_6}\midd
 \tuple{q(A),r(A,B);\red{\id}}^{\ell_4}\\  
 & \to_{\sf unfold} & \tt \tuple{\true;\red{\{A/b,B/b\}}}\midd
 \tuple{r(A,B);\red{\{A/b\}}}^{\ell_6}\midd
 \tuple{q(A),r(A,B);\red{\id}}^{\ell_4}\\  
 & \to_{\sf next} & \tt 
 \tuple{r(A,B);\red{\{A/b\}}}^{\ell_6}\midd
 \tuple{q(A),r(A,B);\red{\id}}^{\ell_4}\\  
 & \to_{\sf unfold} & \tt 
 \tuple{\true;\red{\{A/b,B/c\}}}\midd
 \tuple{q(A),r(A,B);\red{\id}}^{\ell_4}\\  
\end{array}
\]
with computer answer $\tt \{A/b,B/c\}$.
\end{example}
Clearly, the semantics in 
Figure~\ref{fig:deterministic} is deterministic. In the following,
we assume that a fixed program $P$ is considered for stating formal
properties.

\begin{theorem}
  Let $S$ be a state. Then, at most one rule from the semantics
  in Figure~\ref{fig:deterministic} is applicable.	
\end{theorem}

\begin{proof}
  The proof is straightforward since the conditions of the
  rules do not overlap:
  \begin{itemize}
  \item If the leftmost query is not headed  by 
  $\true$ nor $\fail$ and
  the query is not labelled, only rule \textsf{choice} and
  \textsf{choice\_fail} are applicable, and the conditions
  trivially do not overlap.
  \item If the leftmost query is labelled, only rule \textsf{unfold}
  is applicable.
  \item Finally, if the leftmost query is headed by $\fail$
  (resp.\ $\true$) then only rule \textsf{backtrack}
  (resp.\ \textsf{next}) is applicable.
\end{itemize}
\end{proof}
Now, we prove that the deterministic operational semantics is
sound in the sense
that it explores the SLD tree of a query following
Prolog's depth-first search strategy: 

\begin{theorem}
  Let $\tuple{\cA;\id}$ be an initial state. If
  $\tuple{\cA;\id} \to^\ast \tuple{\true;\theta}\midd S$,
  then $\cA \leadsto^\ast_{\theta} \true$, up to variable 
  renaming.
\end{theorem}

\begin{proof}
  Here, we prove a more general claim. Let us consider an arbitrary
  query, $\tuple{\cA;\sigma}$ with $\tuple{\cA;\sigma} \to^\ast
  Q_1\midd\ldots\midd Q_m$, where $Q_i$ is either
  $\tuple{\cB_i;\sigma\theta_i}$ or 
  $\tuple{\cB_i;\sigma\theta_i}^{\ell_i}$, $i=1,\ldots,m$.
  Then, we have $\cA\sigma 
  \leadsto^\ast_{\theta_i} \cB_i\sigma\theta_i$ for
  all $i=1,\ldots,m$ such that $\cB_i \neq (\fail,\cB')$ for
  some $\cB'$, up to variable renaming. We exclude the queries with
  $\fail$ since failures are not made explicit in the definition
  of SLD resolution (this is just a device of our deterministic
  semantics to point out that either
  a backtracking step should be performed next or the derivation
  is failed).
  
  We prove the claim by induction on the number $n$ of steps in
  the former derivation:
  $\tuple{\cA;\sigma} \to^\ast
  Q_1\midd 
  \ldots \midd Q_m$.
  Since the base case ($n=0$) is trivial, let us consider
  the inductive case ($n>0$). Here, we assume a derivation of 
  $n+1$ steps from $\tuple{\cA;\sigma}$. 
  By the induction hypothesis, we have $\cA\sigma 
  \leadsto^\ast_{\theta_i} \cB_i\sigma\theta_i$ for
  all $i=1,\ldots,m$ such that $\cB_i \neq (\fail,\cB')$ for
  some $\cB'$.
  We now distinguish several possibilities depending on the
  applied rule to the state $Q_1\midd\ldots \midd Q_m$:
  \begin{itemize}
  \item If the applied rule is \textsf{backtrack} or \textsf{next},
  we have 
  \[ 
  Q_1\midd Q_2\midd\ldots \midd Q_m 
  \to Q_2\midd\ldots\midd Q_m
  \]
  and the claim trivially holds 
  by the induction hypothesis.
  \item If the applied rule is \textsf{choice}, we have
  \[
  Q_1\midd\ldots \midd Q_m \to Q_1^{\ell_1}\midd \ldots\midd 
  Q_1^{\ell_k}\midd Q_2\midd \ldots \midd Q_m
  \]
  for some $k>0$, and the claim also follows trivially from
  the induction hypothesis.
  \item If the applied rule is \textsf{choice\_fail},
  the claim follows immediately by the induction hypothesis since
  a query of the form $(\fail,\cB')$ is not considered.
  \item Finally, let us consider that the applied rule 
  is \textsf{unfold}. Let $Q_1 = \tuple{A,\cB;\sigma\theta_1}^{\ell_1}$.
  Then, we have
  \[
    \tuple{A,\cB;\sigma\theta_1}^{\ell_1}\midd Q_2\midd\ldots\midd Q_m
    \to
    \tuple{\cB',\cB;\sigma\theta_1\theta'}\midd Q_2\midd\ldots\midd Q_m 
  \]
  if $\clause(\ell_1,P)= H\leftarrow \cB'$ and 
  $\mgu(A\sigma\theta_1,H)=\theta'$. Then, we also have
  an SLD resolution step of the form
  $(A,\cB)\sigma\theta_1 \leadsto_{\theta'} (\cB',\cB)\sigma\theta_1\theta'$ using 
  the same clause\footnote{For simplicity, we assume that 
  the same renamed clauses are considered in both derivations.} 
  and computing 
  the same mgu and, thus, the claim follows from the induction 
  hypothesis.
  \end{itemize}
\end{proof}
Note that the deterministic semantics is sound but 
\emph{incomplete} in general
since it implements a depth-first search strategy.

%%%%%%%%%%%%%%%%%%%%%%%%%%%%%%%%%%%%%%%%%%%%%%%%%%%%%%%%%%%%%%%%%
\subsection{A Reversible Semantics} \label{sec:reversiblesemantics}

Now, we extend the deterministic operational semantics of 
Figure~\ref{fig:deterministic} in order to make it reversible.
Our reversible semantics is defined on \emph{configurations}:

\begin{definition}[configuration]
A \emph{configuration} is defined as a pair
$S \sep \his$ where $S$ is a state (as defined in 
Definition~\ref{def:state}) and $\his$ is a list representing 
the history of the configuration. 
Here, we consider the following history events: 
\begin{itemize}
  \item $\choice(n)$: denotes a choice step with $n$ branches;
  \item $\unf(A,\theta,\ell)$:  represents an unfolding step
  where the selected atom is $A$, the answer computed so far
  is $\theta$, and the selected clause is labelled with $\ell$;
  \item $\fail(A)$: is associated to rule 
  \textsf{choice\_fail} and denotes that the selected atom $A$
  matches no rule;
  \item $\exit(A)$: denotes that the execution of atom $A$
  has been completed (see below);
  \item $\backtrack(\cB,\theta)$: represents a backtracking step,
  where $\tuple{\fail,\cB;\theta}$ is the query that failed;
  \item $\nextsol(\theta)$: denotes an application of rule
  $\nextsol$ after an answer $\theta$ is obtained.   
\end{itemize}
We use Haskell's notation for lists and denote by $s\cons\his$ 
a history with first element $s$ and tail $\his$; an empty
history is denoted by $\nil$.
\end{definition}

\begin{figure}[t]
  \[
  \hspace{-2ex}\begin{array}{r@{~}l}
    \mathsf{(backtrack)} & {\displaystyle 
      \frac{} 
        {\tuple{\fail,\cB;\theta}\midd 
        \tuple{A,\cB';\theta'}\midd S\sep\his 
        \fto 
        \tuple{A,\cB';\theta'}\midd S\sep
        \backtrack(\cB,\theta)
        \cons\his {}_\olive{\redo(A\theta')}}
        }\\[4ex]

    \mathsf{(next)} & {\displaystyle 
      \frac{S\neq\epsilon} 
        {\tuple{\true;\theta}\midd S\sep
        \his {}_{\sf answer(\theta)} \fto S\sep
        \nextsol(\theta)\cons\his}
        }
        \\[4ex]

     \textsf{(choice)} &  {\displaystyle 
      \frac{A\neq\true\wedge A\neq\fail\wedge A\neq\ret(A')\wedge
      \clauses(A\theta,P) = \{\ell_1,\ldots,\ell_m\}\wedge m>0} 
        {\tuple{A,\cB;\theta}\midd S\sep
        \his {}_\green{\call(A\theta)} 
        \fto
          \tuple{A,\cB;\theta}^{\ell_1} \midd 
          \ldots\midd
          \tuple{A,\cB;\theta}^{\ell_m}\midd  S\sep
          \choice(m)\cons\his}
        } \\[4ex]

     \textsf{(choice\_fail)} &  {\displaystyle 
      \frac{A\neq\true\wedge A\neq\fail\wedge A\neq\ret(A')\wedge
      \clauses(A\theta,P) = \emptyset } 
        {\tuple{A,\cB;\theta}\midd S\sep
         \his {}_\green{\call(A\theta)} 
        \fto
          \tuple{\fail,\cB;\theta}\midd  S\sep
          \choicefail(A)\cons\his {}_\red{\fail(A\theta)}}
        } \\[4ex]

    \mathsf{(unfold)} & {\displaystyle 
      \frac{A\neq\ret(A')\wedge \clause(\ell,P) = H\leftarrow B_1,\ldots,B_n
 	\wedge \mgu(A\theta,H)=\sigma}
        {\tuple{A,\cB;\theta}^{\ell}
        \midd S\sep\his
    \fto \tuple{B_1,\ldots,B_n,\ret(A),\cB;\theta\sigma}\midd S\sep
          \unf(A,\theta,\ell)\cons\his}
        } \\[4ex]

    \mathsf{(exit)} & {\displaystyle 
      \frac{}
        {\tuple{\ret(A),\cB;\theta}
        \midd S\sep\his {}_\green{\exit(A\theta)}
    \fto \tuple{\cB;\theta}
    \midd S\sep \exit(A)\cons\his}
    } \\ %%[3ex]
    \end{array}
    \]
  %\rule{\linewidth}{1pt}
  \caption{Forward reversible semantics} 
  \label{fig:forward}
\end{figure}

The reversible (forward) semantics is shown in 
Figure~\ref{fig:forward}.\footnote{The subscripts of some
configurations: \textsf{call}, \textsf{exit}, \textsf{fail}, 
\textsf{redo}, and \textsf{answer}, can be ignored for now.
They will become useful in the next section.} 
The rules of the reversible semantics are 
basically self-explanatory. They are essentially the same as in the
standard deterministic semantics of Figure~\ref{fig:deterministic}
except for the following differences:
\begin{itemize}
\item First, configurations now keep a \emph{history}
with enough information for undoing the steps of a computation.
\item And, secondly, unfolding an atom $A$ now
adds a new call of the form $\ret(A)$ after the atoms of the
body (if any) of the considered program clause. 
This is then used in rule \textsf{exit} in order to determine when the 
call has been completed successfully ($\ret(A)$ marks the exit of 
a program clause). This extension is not introduced for reversibility,
but it is part of the design of our reversible debugger 
(see Section~\ref{sec:rever}, where the reversible debugger 
\textsf{rever} is presented). Here, and in the following, 
we assume that programs contain no
predicate named $\ret/1$.
\end{itemize}
We note that extending our developments to SLD resolution with an
arbitrary computation rule (i.e., different from Prolog's rule, which
always selects the leftmost atom) is not difficult. Basically, one
would only need to extend the $\unf$ elements as follows: 
$\unf(A,\theta,i,\ell)$,
where $i$ is the position of the selected atom in the current
query.

\begin{example} \label{ex:running3}
Consider again the program of Example~\ref{ex:running} and the
initial query: $\tt \tuple{p(X,Y);\id}\sep\nil$. In order to 
reach the first  computed answer, $\tt \{A/b,B/b\}$, we
perform the derivation shown in Figure~\ref{fig:example}.

\begin{figure}[p]
$
\begin{array}{l}
 \tt \tuple{p(A,B);\id}\sep\red{\nil}\\

  \fto_{\sf choice}  \tt \tuple{p(A,B);\id}^{\ell_1}
   \sep \red{[\choice(1)]}  \\
   
  \fto_{\sf unfold}  \tt \tuple{q(A),r(A,B),\blue{ret(p(A,B))};\id}
   \sep \red{[\unf(p(A,B),\id,\ell_1),\choice(1)]}\\  

  \fto_{\sf choice}  \tt \tuple{q(A),r(A,B),\blue{ret(p(A,B))};\id}^{\ell_2}
   \midd \tuple{q(A),r(A,B),\blue{ret(p(A,B))};\id}^{\ell_3}\\
   \hspace{9ex}\tt \midd
 \tuple{q(A),r(A,B),\blue{ret(p(A,B))};\id}^{\ell_4}
    \sep \red{[\choice(3),\unf(p(A,B),\id,\ell_1),\choice(1)]}\\  

  \fto_{\sf unfold}  \tt \tuple{\blue{ret(q(A))},r(A,B),\blue{ret(p(A,B))};\{A/a\}}
 \midd \tuple{q(A),r(A,B),\blue{ret(p(A,B))};\id}^{\ell_3}\\
 \hspace{9ex}\tt \midd
 \tuple{q(A),r(A,B),\blue{ret(p(A,B))};\id}^{\ell_4}\\  
  \hspace{9ex} \tt \sep\: \red{[\unf(q(A),\id,\ell_2),\choice(3),\unf(p(A,B),\id,\ell_1),\choice(1)]}\\  
  
  \fto_{\sf exit}  \tt \tuple{r(A,B),\blue{ret(p(A,B))};\{A/a\}}
 \midd \tuple{q(A),r(A,B),\blue{ret(p(A,B))};\id}^{\ell_3}\\
 \hspace{9ex}\tt \midd
 \tuple{q(A),r(A,B),\blue{ret(p(A,B))};\id}^{\ell_4}\\  
  \hspace{9ex} \tt \sep\: \red{[\exit(q(A)),
  \unf(q(A),\id,\ell_2),\choice(3),\unf(p(A,B),\id,\ell_1),\choice(1)]}\\  
  
  \fto_{\sf choice\_fail}  \tt \tuple{\fail,\blue{ret(p(A,B))};\{A/a\}}
 \midd \tuple{q(A),r(A,B),\blue{ret(p(A,B))};\id}^{\ell_3}\\
 \hspace{9ex}\tt \midd
 \tuple{q(A),r(A,B),\blue{ret(p(A,B))};\id}^{\ell_4}\\  
  \hspace{9ex} \tt \sep\: \red{[\fail(r(A,B)),\exit(q(A)),
  \unf(q(A),\id,\ell_2),\choice(3),\unf(p(A,B),\id,\ell_1),\choice(1)]}\\  
  
  \fto_{\sf backtrack}  \tt 
  \tuple{q(A),r(A,B),\blue{ret(p(A,B))};\id}^{\ell_3}\midd
 \tuple{q(A),r(A,B),\blue{ret(p(A,B))};\id}^{\ell_4}\\  
  \hspace{9ex} \tt \sep\: \red{[\backtrack(ret(p(A,B)),\{A/a\}),\fail(r(A,B)),\exit(q(A)),
  \unf(q(A),\id,\ell_2),\choice(3),}\\
   \hspace{11ex} \tt \red{\unf(p(A,B),\id,\ell_1),\choice(1)]}\\  
  
  \fto_{\sf unfold}  \tt 
  \tuple{\blue{\ret(q(A))},r(A,B),\blue{ret(p(A,B))};\{A/b\}} \midd
 \tuple{q(A),r(A,B),\blue{ret(p(A,B))};\id}^{\ell_4}\\  
  \hspace{9ex} \tt \sep\: \red{[\unf(q(A),\id,\ell_3),\backtrack(ret(p(A,B)),\{A/a\}),\fail(r(A,B)),\exit(q(A)),}\\
  \hspace{11ex} \tt \red{\unf(q(A),\id,\ell_2),\choice(3),
   \unf(p(A,B),\id,\ell_1),\choice(1)]}\\  
  
  \fto_{\sf exit}  \tt 
  \tuple{r(A,B),\blue{ret(p(A,B))};\{A/b\}}\midd
 \tuple{q(A),r(A,B),\blue{ret(p(A,B))};\id}^{\ell_4}\\  
  \hspace{9ex} \tt \sep\: \red{[\exit(q(A)),\unf(q(A),\id,\ell_3),\backtrack(ret(p(A,B)),\{A/a\}),\fail(r(A,B)),\exit(q(A)),}\\
  \hspace{11ex} \tt \red{\unf(q(A),\id,\ell_2),\choice(3),
   \unf(p(A,B),\id,\ell_1),\choice(1)]}\\  

  \fto_{\sf choice}  \tt 
  \tuple{r(A,B),\blue{ret(p(A,B))};\{A/b\}}^{\ell_5}
  \midd \tuple{r(A,B),\blue{ret(p(A,B))};\{A/b\}}^{\ell_6}\\
 \hspace{9ex}\tt \midd
 \tuple{q(A),r(A,B),\blue{ret(p(A,B))};\id}^{\ell_4}\\  
  \hspace{9ex} \tt \sep\: \red{[\choice(2),\exit(q(A)),\unf(q(A),\id,\ell_3),\backtrack(ret(p(A,B)),\{A/a\}),\fail(r(A,B)),}\\
  \hspace{11ex} \tt \red{\exit(q(A)),\unf(q(A),\id,\ell_2),\choice(3),
   \unf(p(A,B),\id,\ell_1),\choice(1)]}\\  

  \fto_{\sf unfold}  \tt 
  \tuple{\blue{ret(r(A,B))},\blue{ret(p(A,B))};\{A/b,B/b\}}
  \midd \tuple{r(A,B),\blue{ret(p(A,B))};\{A/b\}}^{\ell_6}\\
 \hspace{9ex}\tt \midd
 \tuple{q(A),r(A,B),\blue{ret(p(A,B))};\id}^{\ell_4}\\  
  \hspace{9ex} \tt \sep\: \red{[\unf(r(A,B),\{A/b\},\ell_5),\choice(2),\exit(q(A)),\unf(q(A),\id,\ell_3),}\\
  \hspace{11ex} \tt \red{\backtrack(ret(p(A,B)),\{A/a\}),\fail(r(A,B)),\exit(q(A)),\unf(q(A),\id,\ell_2),\choice(3),}\\
   \hspace{11ex} \tt \red{\unf(p(A,B),\id,\ell_1),\choice(1)]}\\  
  
  \fto_{\sf exit}  \tt 
  \tuple{\blue{ret(p(A,B))};\{A/b,B/b\}}
  \midd \tuple{r(A,B),\blue{ret(p(A,B))};\{A/b\}}^{\ell_6}\\
 \hspace{9ex}\tt \midd
 \tuple{q(A),r(A,B),\blue{ret(p(A,B))};\id}^{\ell_4}\\  
  \hspace{9ex} \tt \sep\: \red{[\exit(r(A,B)),\unf(r(A,B),\{A/b\},\ell_5),\choice(2),\exit(q(A)),\unf(q(A),\id,\ell_3),}\\
  \hspace{11ex} \tt \red{\backtrack(ret(p(A,B)),\{A/a\}),\fail(r(A,B)),\exit(q(A)),\unf(q(A),\id,\ell_2),\choice(3),}\\
   \hspace{11ex} \tt \red{\unf(p(A,B),\id,\ell_1),\choice(1)]}\\  

  \fto_{\sf exit}  \tt 
  \tuple{\true;\{A/b,B/b\}}
  \midd \tuple{r(A,B),\blue{ret(p(A,B))};\{A/b\}}^{\ell_6} \midd
 \tuple{q(A),r(A,B),\blue{ret(p(A,B))};\id}^{\ell_4}\\  
  \hspace{9ex} \tt \sep\: \red{[\exit(p(A,B)),\exit(r(A,B)),\unf(r(A,B),\{A/b\},\ell_5),\choice(2),\exit(q(A)),}\\
  \hspace{11ex} \tt \red{\unf(q(A),\id,\ell_3),\backtrack(ret(p(A,B)),\{A/a\}),\fail(r(A,B)),\exit(q(A)),}\\
   \hspace{11ex} \tt \red{\unf(q(A),\id,\ell_2),\choice(3),\unf(p(A,B),\id,\ell_1),\choice(1)]}\\  

\end{array}
$
\caption{Example derivation with the reversible (forward) semantics.} \label{fig:example}
\end{figure}
\end{example}
It is worthwhile to observe that the drawbacks of \cite{Vid20}
mentioned before are now overcome by using substitutions with
the answer computed so far, together with a deterministic semantics 
where backtracking is dealt with explicitly.

Trivially, the instrumented semantics of 
Figure~\ref{fig:forward} is a conservative extension of
the deterministic semantics of Figure~\ref{fig:deterministic}
since the rules impose no additional condition. The only
difference is the addition of atoms $\ret(A)$ that mark
the exit of a program clause. In the following, given two
states, $S,S'$, we let $S\sim S'$ if they are equal after
removing all atoms of the form $\ret(A)$.

\begin{theorem}
  Let $Q$ be an initial state. Then, $Q \to^\ast S$ iff
  $Q\sep\nil \fto^\ast S'\sep\his$ such that $S\sim S'$
  for some history $\his$, 
  up to variable renaming.
\end{theorem}

%%%%%%%%%%%%%%%%%%%%%%%%%%%%%%%%%%%%%%%%%%%%%%%%%%%%%%%
%\subsection{Backward Reversible Semantics}

Let us now consider backward steps. Here, our goal is to
be able to explore a given derivation backwards. For this
purpose, we introduce a  backward 
operational semantics that is essentially obtained
by switching the configurations in each rule of the
forward semantics, and removing all 
unnecessary premises. The resulting backward semantics
is shown in Figure~\ref{fig:backward}.
Let us just add that, in rule $\ol{\sf unfold}$, we
use the auxiliary function $\body(\ell,P)$ to denote
the body of clause labelled with $\ell$ in program $P$,
and, thus, $|\body(\ell,P)|$ represents the number
of atoms in the body of this clause.\footnote{
As is common, $|S|$ denotes the cardinality of the set 
or sequence $S$.
} This information
was stored explicitly in our previous approach  
\cite{Vid20}.

\begin{figure}[t]
  \[
  \hspace{-2ex}\begin{array}{r@{~~~}l}
    \mathsf{(\ol{backtrack})} & 
        S\sep
        \backtrack(\cB,\theta)
        \cons\his % {}_\olive{\redo(A\theta')}
        \bto     
        \tuple{\fail,\cB;\theta}\midd 
        S\sep\his 
        \\[2ex]

    \mathsf{(\ol{next})} & 
    S\sep \nextsol(\theta)\cons\his
        \bto
        \tuple{\true;\theta}\midd S\sep
        \his % {}_{\sf answer(\theta)} 
        \\[2ex]

     \textsf{(\underline{choice})} &  
          \tuple{A,\cB;\theta}^{\ell_1} \midd 
          \ldots\midd
          \tuple{A,\cB;\theta}^{\ell_m}\midd  S\sep
          \choice(m)\cons\his
	\bto
     \tuple{A,\cB;\theta}\midd S\sep
        \his % {}_\green{\call(A\theta)} 
      \\[2ex]

     \textsf{(\underline{choice\_fail})} &  
      \tuple{\fail,\cB;\theta}\midd  S\sep
          \choicefail(A)\cons\his %{}_\red{\fail(A\theta)}}
     \bto
     \tuple{A,\cB;\theta}\midd S\sep
         \his %{}_\green{\call(A\theta)} 
       \\[2ex]

    \mathsf{(\underline{unfold})} & 
       \tuple{B_1,\ldots,B_n,\ret(A),\cB;\theta\sigma}\midd S\sep
          \unf(A,\theta,\ell)\cons\his
        \bto
        \tuple{A,\cB;\theta}^{\ell}
        \midd S\sep\his
        \\[.7ex]
        & \mbox{where}~|\body(\ell,P)| = n 
       \\[2ex]

    \mathsf{(\underline{exit})} & 
    \tuple{\cB;\theta}
    \midd S\sep \exit(A)\cons\his
    \bto
      \tuple{\ret(A),\cB;\theta}
        \midd S\sep\his % {}_\green{\exit(A\theta)}
    \\ 
    \end{array}
    \]
  \caption{Backward reversible semantics} 
  \label{fig:backward}
\end{figure}

\begin{example} \label{ex:running4}
  If we consider the configurations of Figure~\ref{fig:example} from
  bottom to top, they constitute a backward derivation using the
  rules of Figure~\ref{fig:backward}.
\end{example}
The following result states the reversibility of our semantics:

\begin{lemma}
  Let $\cC,\cC'$ be configurations. If $\cC \fto \cC'$, then $\cC'\bto \cC$,
  up to variable renaming.
\end{lemma}

\begin{proof}
  The claim follows by a simple case distinction on the applied rule and
  the fact that the backward semantics of Figure~\ref{fig:backward} is
  trivially deterministic since each rule requires a different element on
  the top of the history.
\end{proof}
In principle, one could also prove the opposite direction, i.e., that
$\cC'\bto \cC$ implies $\cC \fto \cC'$, by requiring that $\cC'$ is not
an arbitrary configuration but a ``legal'' one, i.e., a configuration that
is reachable by a forward derivation starting from some initial
configuration.

The above result could be straightforwardly extended to derivations as follows:

\begin{theorem}
  Let $\cC,\cC'$ be configurations. If $\cC \fto^\ast \cC'$, then $\cC'\bto^\ast \cC$,
  up to variable renaming.  
\end{theorem}

%%%%%%%%%%%%%%%%%%%%%%%%%%%%%%%%%%%%%%%%%%%%%%%%%%%%%%%%%%%
\section{A Reversible Debugger for Prolog}
\label{sec:rever}

In this section, we present the design of a reversible debugger for 
Prolog. It is based on the standard 4-port tracer introduced by Byrd
\cite{Byr80,CM94}. The ports are $\call$ (an atom is called), $\exit$ (a call is 
successfully completed), $\redo$ (backtracking requires trying again
some call), and $\fail$ (an atom matches no clause). In contrast to
standard debuggers that can only explore a computation forward, our
reversible debugger allows us to move back and forth.

The implemented debugger, \textsf{rever}, is publicly
available from \url{https://github.com/mistupv/rever}. It 
can be used in two modes:
\begin{itemize}
\item \textit{Debug mode}. In this case, execution proceeds silently
(no information is shown) until the execution of a special predicate
$\tt rtrace/0$ is reached (if any).  The user can include a call to 
this predicate in the source program in order to start tracing the
computation (i.e., it behaves as \texttt{trace/0} in most Prolog systems). 
Tracing also starts if an \emph{exception} is produced during the evaluation
of a query.
This mode is invoked with a call of the form $\tt rdebug(query)$, where
$\tt query$ is the initial query whose execution we want to explore.

\item \textit{Trace mode}. In this mode, port information is shown from
the beginning. One can invoke the trace mode with $\tt rtrace(query)$. 
Note that it is equivalent to calling $\tt rdebug((rtrace,query))$.
\end{itemize}
Our reversible debugger essentially implements the transition rules in 
Figures~\ref{fig:forward} and \ref{fig:backward}. As the reader may
have noticed,  some configurations in Figure~\ref{fig:forward} are labeled with a 
subscript: it denotes the output of a given port. Moreover, there is
an additional label in rule $\nextsol$ which denotes that, at this
point, an answer must be shown to the user. 

In tracing mode,  every time that a configuration with a subscript is reached, 
the execution stops, shows the corresponding port information,
and waits for the user to press some key. We basically consider the following 
keys: $\downarrow$ (or Enter) proceeds with the next (forward) step;
$\uparrow$ performs a backward step; $s$ (for skip) shows the port
information without waiting for any user interaction; $t$ enters the tracing mode; 
$q$ quits the debugging session.

For instance, given the initial call $\tt rtrace(p(A,B))$, and according to
the forward derivation shown in Figure~\ref{fig:example}, our debugger
displays the sequence shown in Figure~\ref{fig:rever} (a).
Now, if one presses ``$\uparrow$" repeatedly, the sequence displayed in 
Figure~\ref{fig:rever} (b) is shown. Note that ports are prefixed by
the symbol ``$\:\hat{~}\:$'' in backward derivations. Of course, the user can move
freely back and forth.

\begin{figure}[t]
\centering
$
\begin{tabular}{l@{~~~~~~}l}
\tt \green{Call:} p(A,B) &  \tt $\hat{~}$\green{Exit:} p(b,b)\\
\tt \green{Call:} q(A) &  \tt $\hat{~}$\green{Exit:} r(b,b)\\
\tt \green{Exit:} q(a) &  \tt $\hat{~}$\green{Call:} r(b,B)\\
\tt \green{Call:} r(a,B) &  \tt $\hat{~}$\green{Exit:} q(b)\\
\tt \red{Fail:} r(a,B) &  \tt $\hat{~}$\olive{Redo:} q(A)\\
\tt \olive{Redo:} q(A) &  \tt $\hat{~}$\red{Fail:} r(a,B)\\
\tt \green{Exit:} q(b) &  \tt $\hat{~}$\green{Call:} r(a,B)\\
\tt \green{Call:} r(b,B) &  \tt $\hat{~}$\green{Exit:} q(a)\\
\tt \green{Exit:} r(b,b) &  \tt $\hat{~}$\green{Call:} q(A)\\
\tt \green{Exit:} p(b,b) &  \tt $\hat{~}$\green{Call:} p(A,B)\\
\tt **Answer: A = b, B = b\\[1ex]
\hspace{8ex}(a) & \hspace{8ex}(b)
\end{tabular}
$
\caption{Trace Example with \textsf{rever}} \label{fig:rever}
\end{figure}

Reversible debugging might be especially useful when we have an
execution that produces some exception at the end. 
With our tool, one can easily inspect the execution backwards 
from the final state that produced the error. 

Let us mention that,  in order to avoid the use of a ground representation and
having to implement all basic operations (mgu, substitution application
and composition, etc), substitutions are represented in its
equational form. E.g., substitution  $\tt \{A/a,B/b\}$
is represented by $\tt A=a,B=b$. This equational representation of a
mgu can be easily obtained by using the predefined predicate 
$\tt unify/3$.  This representation is much more efficient than
storing pairs of atoms (as in \cite{Vid20}), that must be unified
once and again at each execution step.  

Finally, let us mention that, despite the simplicity of the implemented 
system (some 500 lines of code in SWI Prolog), our debugger
is able to deal with medium-sized programs (e.g., it has been used to
debug the debugger itself).

%%%%%%%%%%%%%%%%%%%%%%%%%%%%%%%%%%%%%%%%%%%%%%%%%%%%%%%%%%%%%%
\section{Related Work}
\label{sec:relwork}

The closest approach is clearly the preliminary version of this work in \cite{Vid20}.
There are, however, several significant differences: \cite{Vid20} presents
a reversible version of the usual, nondeterministic SLD resolution. Therefore,
backtracking steps cannot be undone. This is improved in this paper by
considering a deterministic semantics that models the traversal of the
complete SLD tree. Moreover, \cite{Vid20} considers a simple but 
very inefficient representation for the history, which is greatly improved
in this paper. Finally, we provide proofs of some formal properties for
our reversible semantics, as well as a publicly available implementation 
of the debugger, the system \textsf{rever}.

Another close approach we are aware of is that of Opium \cite{Duc99},
which introduces a trace query language for inspecting and analyzing
trace histories. In this tool, the trace history of the considered
execution is stored in a database, which is then used for trace
querying. Several analysis can then be defined in Prolog itself by
using a set of given primitives to explore the trace elements.
In contrast to our approach, Opium is basically a so-called
``post-mortem" debugger that allows one to analyze the 
trace of an execution. Therefore, the goal is different from that
of this paper.

%%%%%%%%%%%%%%%%%%%%%%%%%%%%%%%%%%%%%%%%%%%%%%%%%%%%%%%%%%%
\section{Concluding Remarks and Future Work}
\label{sec:conclusion}

We have proposed a novel reversible debugging scheme for
logic programs by defining an appropriate Landauer embedding
for a deterministic operational semantics.  Essentially, the states
of the semantics are extended with a \emph{history} that keeps
track of all the information which is needed to be able to undo the
steps of a computation. We have proved a number of formal
properties for our reversible semantics.
Moreover, the ideas have been put
into practice in the reversible debugger \textsf{rever}, which is
publicly available  from \url{https://github.com/mistupv/rever}. 
Our preliminary experiments with the debugger
have shown promising results.

As for future work, we are currently working on extending the 
debugger in order to cope with negation and the cut.   Also, we plan
to define a more compact representation for the history, so that
it can scale up better to larger programs and derivations.

\subsection*{Acknowledgements}

The author gratefully acknowledges the editors, John Gallagher,
Roberto Giacobazzi and Pedro L\'opez-Garc\'{\i}a, for the opportunity
to contribute to this volume, dedicated to Manuel Hermenegildo on
the occasion of his 60th birthday.

%%%%%%%%%%%%%%%%%%%%%%%%%%%%%%%%%%%%%%%%%%%%%%%%%%%%%%%%%%%
%%%%%%%%%%%%%%%%%%%%%%%%%%%%%%%%%%%%%%%%%%%%%%%%%%%%%%%%%%%
\bibliographystyle{splncs04}
%\bibliography{biblio}

\end{document}